\documentclass[11pt, a4paper, oneside]{article}

\usepackage[letterpaper, left=2.3cm, right=2.3cm, top=2cm, bottom=2cm]{geometry}%
\usepackage[english]{babel}
\usepackage[utf8x]{inputenc}
\usepackage[T1]{fontenc}

\usepackage{bm}
\usepackage[breaklinks]{hyperref}

\usepackage{enumerate}
\usepackage[numbers]{natbib}
\usepackage{mathtools}
\usepackage{relsize}

\usepackage{graphicx,adjustbox}
\usepackage[dvipsnames]{xcolor}

\usepackage{amsfonts,amssymb,amsthm}
\usepackage{microtype}
\usepackage{algorithm}
\usepackage{algpseudocode}
\usepackage{algorithmicx}
\usepackage{bbm}
\usepackage{enumerate}
\usepackage{mathtools}
\mathtoolsset{showonlyrefs}

\usepackage{tikz}
\usetikzlibrary{shapes.geometric,backgrounds,calc}
\tikzset{>=latex}
\usetikzlibrary{arrows,shapes,tikzmark,shadows, fit,petri,%
                topaths}

\usepackage{hyperref}

\usepackage{thmtools}
\usepackage{thm-restate}
\usepackage{harpoon}

\newcommand{\nequiv}{\not\equiv}

\newcommand{\finalTimer}{\mathcal{T}}
\newcommand{\Timer}{\tau}

\newcommand{\adv}[1]{{#1}^{*}}
\newcommand{\OPT}{\mathcal{OPT}}
\newcommand{\ALG}{\A}
\newcommand{\Location}{\ell}
\newcommand{\M}{\mathcal{M}}
\newcommand{\A}{\mathcal{ALG}}

\newcommand{\Reals}{\mathbb{R}}

\newcommand{\Expect}{\mathbb{E}}

\newcommand{\sCost}{\mathrm{cost}^{\mathit{space}}}
\newcommand{\tCost}{\mathrm{cost}^{\mathit{time}}}
\newcommand{\Cost}{\mathrm{cost}}
\newcommand{\Leaves}{L}
\newcommand{\LCA}{\mathrm{lca}}

\newcommand{\Ancestors}{\mathit{anc}}
\newcommand{\Descendant}{\mathit{des}}

\newcommand{\Active}{\mathit{C}}

\newcommand{\OTwo}{$O_2$}
\newcommand{\OD}{$O_D$}
\newcommand{\ddelta}{$\Delta_D$}
\newcommand{\ndelta}{$\Delta_n$}
\newcommand{\hdelta}{$\Delta_H$}
\newcommand{\SK}{$\mathcal{S}_K$}
\newcommand{\SH}{$\mathcal{S}_H$}
\newcommand{\SHstar}{$\mathcal{S}_H^*$}
\newcommand{\Per}{$\Pi$}

\newtheorem*{axiomOTwo}{$O_2$}
\newtheorem*{axiomOD}{$O_D$}
\newtheorem*{axiomPi}{$\Pi$}
\newtheorem*{axiomDeltaD}{$\Delta_D$}
\newtheorem*{axiomDeltaN}{$\Delta_n$}
\newtheorem*{axiomDeltaH}{$\Delta_H$}
\newtheorem*{axiomSK}{$\mathcal{S}_K$}
\newtheorem*{axiomSHstar}{$\mathcal{S}_H^*$}
\newtheorem*{axiomSH}{$\mathcal{S}_H$}

\newtheorem{observation}{Observation}
\newtheorem{theorem}{Theorem}
\newtheorem{definition}{Definition}
\newtheorem{lemma}{Lemma}

\newtheorem{example}{Example}

\newcommand\mydots{\makebox[1em][c]{.\hfil.\hfil.}}

\newcommand\clock[4][2]{%
\begin{tikzpicture}[scale=#1,line cap=round,line width=#1*4pt]
\node[scale=#1] (A) at (#4) {};
\filldraw [fill=yellow!20] (A) circle (2cm);
\foreach \angle / \label in
{0/3, 30/2, 60/1, 90/12, 120/11, 150/10, 180/9,
210/8, 240/7, 270/6, 300/5, 330/4}
{
\draw[line width=#1*1pt] ($(A) + (\angle:1.8cm)$) -- ($(A) + (\angle:2cm)$);
}
\foreach \angle in {0,90,180,270}
\draw[line width=#1*2pt] ($(A) + (\angle:1.6cm)$) -- ($(A) + (\angle:2cm)$);
\draw[rotate=90,line width=#1*6pt] (A) -- ($(A) + (-#2*30-#3*30/60:1cm)$); 
\draw[rotate=90,line width=#1*4pt] (A) -- ($(A) + (-#3*6:1.3cm)$);
\path [fill=black] (A) circle (5pt);
\end{tikzpicture}
}

\tikzfading[name=arrowfading, top color=transparent!0, bottom color=transparent!95]
\tikzset{arrowfill/.style={#1,general shadow={fill=black, path fading=arrowfading}}}
\tikzset{arrowstyle/.style n args={3}{draw=#2,arrowfill={#3}, single arrow,minimum height=#1, single arrow,single arrow head extend=.1cm,}}

\NewDocumentCommand{\tikzfancyarrow}{O{2cm} O{black} O{top color=black!20, bottom color=white} m}{
\tikz[baseline=-0.5ex]\node [arrowstyle={#1}{#2}{#3}] {#4};
}

\newcommand{\RN}[1]{%
  \textup{\uppercase\expandafter{\romannumeral#1}}%
}

\usepackage{newfloat}
\DeclareFloatingEnvironment{algo}

\algdef{SE}[SUBALG]{Indent}{EndIndent}{}{\algorithmicend\ }%
\algtext*{Indent}
\algtext*{EndIndent}

\tikzset{cross/.style={cross out, thick, draw=black, minimum size=2*(#1-\pgflinewidth), inner sep=0pt, outer sep=0pt},
cross/.default={5pt}}

\title{Online $k$-Way Matching with Delays and the $H$-Metric}
\author{
  Darya Melnyk\\
  \small\texttt{darya.melnyk@aalto.fi}\\
  \small ETH Z\"urich\\
  \small Aalto University
  \and
  Yuyi Wang \\
  \small\texttt{yuwang@ethz.ch} \\  
  \small ETH Z\"urich
  \and
  Roger Wattenhofer \\ 
  \small\texttt{wattenhofer@ethz.ch} \\  
  \small ETH Z\"urich
}
\date{}

\begin{document}

\maketitle
\thispagestyle{empty}
\begin{abstract}
In this paper, we study $k$-Way Min-cost Perfect Matching with Delays -- the $k$-MPMD problem. This problem considers a metric space with $n$ nodes. Requests arrive at these nodes in an online fashion. The task is to match these requests into sets of exactly $k$, such that space and time cost of all matched requests are minimized. The notion of the space cost requires a definition of an underlying metric space that gives distances of subsets of $k$ elements. For $k>2$, the task of finding a suitable metric space is at the core of our problem: We show that for some known generalizations to $k=3$ points, such as the $2$-metric~\cite{Gahler1963} and the $D$-metric~\cite{Dmetric}, there exists no competitive randomized algorithm for the $3$-MPMD problem. The $G$-metrics~\cite{mustafa2006} are defined for 3 points and allows for a competitive algorithm for the $3$-MPMD problem. For $k>3$ points, there exist two generalizations of the $G$-metrics known as $n$- and $K$-metrics~\cite{SamAsNMetric,KamAlKhan2012}. We show that neither the $n$-metrics nor the $K$-metrics can be used for the $k$-MPMD problem. On the positive side, we introduce the $H$-metrics, the first metrics to allow for a solution of the $k$-MPMD problem for all $k$.  
In order to devise an online algorithm for the $k$-MPMD problem on the $H$-metrics, we embed the $H$-metric into trees with an $O(\log n)$ distortion. 
Based on this embedding result, we extend the algorithm proposed by \citet{AzarCK2017} and achieve a competitive ratio of $O(\log n)$ for the $k$-MPMD problem. 

\end{abstract}

\section{Introduction}

With annual revenue in the order of 100 billion dollars, the gaming industry is about three times bigger than the movie industry. At its core, there are several online gaming platforms such as Xbox Live, Playstation Network, Steam, UPlay, QQ Games, and soon Google Stadia. 
Most games played online are multi-player games. The number of players per game is a vital parameter $k$ of a game. Games for pretty much any value of $k$ exist, e.g., $k=2$ players (Chess, Go
), $k=3$ players (Dou dizhu), $k=4$ players (Bridge
), $k=5$ (Dota 2), $k=8$ (Dirt 4), $k=24$ (Forza), 
$k=100$ (Fortnite). 

One of the main tasks of any online platform is to match an arriving player to $k-1$ opponents. The gaming platform must therefore optimize two conflicting goals. First, the $k$ players that are matched with each other should be similar regarding various characteristics, e.g., similar playing strength, similar geographic region, similar hardware. On the other hand, players do not want to wait long before they are matched, so a gaming platform must strive to match any player quickly. 

The special case of $k=2$ was first studied by \citet{EmekKW2016}. They introduced the online \emph{Min-cost Perfect Matching with Delays (MPMD)} problem: Given an online sequence of arriving players, MPMD minimizes (i) the matching cost between the two matched players and (ii) the delay incurred by the players waiting to be matched. \citet{EmekKW2016} provided a first polylog-competitive algorithm, which was later improved to $O(\log n)$ by \citet{AzarCK2017}, where $n$ is the number of points in the finite metric space that is used to model the matching cost between two players. 

The MPMD problem in \cite{EmekKW2016,AzarCK2017} is restricted to two players. Two-player games do however only constitute a small fraction of the gaming market. As \citet{AzarCK2017} point out, it is natural to ask whether one can achieve the $O(\log n)$ bound for games with $k$ players: 

\emph{``[An] interesting problem to pursue is
the problem of min-cost $k$-way matching,
where the goal is to partition the requests into sets
of size $k$. We need to identify interesting constraints
on the connection cost, which generalize the metric
properties, and which admit a competitive algorithm.''}

Our paper analyzes this extension to $k$ players, where $k\geq 3$. We formally call the problem \emph{$k$-way Min-cost Perfect Matching with Delays: $k$-MPMD}.
When generalizing the problem to $k$ players, goals (i) and (ii) also need to be generalized. Since the second goal just measures the total delay, it does not change for a larger $k$. 
Generalizing the first goal requires us to define a cost for $k$-way matching. Such a matching can only be applied to a metric space that defines distances for any subset of $k$ points from the metric. Surprisingly, this generalization turns out to be non-trivial. While it is straightforward to define metrics for $k=2$, there exist many ways of generalizing a metric to $k>2$ players. Our results show that all known generalized metrics on $k$ points (e.g. $2$-, $n$- or $K$-metrics) are not suited for the $k$-MPMD problem as their competitive ratios are unbounded. The only set of axioms for which there exists a competitive algorithm are the $G$-metrics, which is only defined for $k=3$ players. The core of this paper is therefore to find a proper definition of generalized $G$-metrics to $k$ points, which does not heavily restrict the metric space. 

Our paper is organized as follows: In Section \ref{sec:generalized_metrics}, we will discuss known generalizations of metrics to three or more points, e.g., $2$-, $G$- and $K$-metrics. In Section \ref{sec:Impossibility2}, we show that $2$- and $D$- metrics fail to be competitive for the $3$-MPMD problem. In Section \ref{sec:impossibilityKNMetric}, we use a similar analysis structure in order to show that there is no competitive randomized algorithm for $n$- and $K$-metrics (generalizations of the $G$-metrics) for the $k$-MPMD problem with $k>3$. The presented counterexamples help us to define restrictions on generalized $G$-metrics. In Section \ref{sec:new_G_k_metric}, we present our novel $H$-metric, which can be shown to be competitive for any $k$. We call it the $H$-metric because it is more general than the $G$-metrics, so alphabetically sitting between $G$ and $K$, close to $G$. 
In Section \ref{sec:reduction} we then reduce the $H$-metrics to a metric that is defined on pairwise distances of points. Finally, in Section \ref{sec:matching_alg} we use the reduction in order to extend the algorithm proposed by \cite{AzarCK2017} to our $H$-metrics and achieve a $O(\log n)$ competitive ratio for the online $k$-MPMD problem.

\section{Related Work}
\label{sec:relatedwork}
Offline matching has become a classic combinatorial problem since the seminal work of \citet{Edmonds1965a,Edmonds1965b}. 
One may argue that in today’s world online matching is practically more relevant than its offline counterpart. 
Online matching algorithms need to deal with continuously arriving input and deliver quality matches on the fly.
Many papers have studied the online matching problem extensively, see for example
\cite{AggarwalGKM2011,BansalBGN2014,BirnbaumM2008,DevanurJK2013,GoelM2008,KalyanasundaramP1993,KarpVV1990,KhullerMV1994,Mehta2013,MehtaSVV2005,MeyersonNP2006,Miyazaki2014,NaorW2015}. 
All these papers do however assume one side of a \textit{bipartite} input to be available to the algorithm offline.
While this online/offline hybrid may perfectly model some applications, many real-world matching problems do not have one side of the data stored initially. 

\citet{EmekKW2016} were the first to study a version of online matching where all input data arrives online, calling the corresponding problem the $2$-MPMD problem. The algorithm in their paper does not decide on a newly arriving request immediately, as only delayed decisions allow competitive algorithms. 
The authors present a randomized algorithm with competitive ratio $O(\log^2 n+ \log \Delta)$ where $\Delta$ is the aspect ratio of this metric.
\citet{AzarCK2017} later improved this result to an $O(\log n)$-competitive randomized algorithm, where the competitive ratio does not depend on the metric space. \citet{Yuyi:ImpatientOnlineMatching} modified the delay function in the MPMD problem to capture convex instead of linear delays.

A first lower bound of $\Omega(\sqrt{\log{n}})$ on the competitive ratio of any randomized algorithm was shown by \citet{AzarCK2017}.
\citet{Ashlagi:BipartiteCase} improved this lower bound to $\Omega\left(\frac{\log{n}}{\log\log{n}}\right)$ such that it almost matches the upper bound. They further consider the MBPMD problem, which is a bipartite version of the $2$-MPMD problem. In this problem, requests arrive at one of two given classes and only requests between classes can be matched. They also provide an $O(\log{n})$-competitive randomized algorithm for the MBPMD problem.

Other recent papers have picked up the $2$-MPMD problem in the deterministic setting. \citet{Bienkowski:DeterministicMPMD} proposed the first deterministic algorithm for the MBPMD problem on general metrics and achieved a competitive ratio of $O\left(m^{2.46}\right)$, where $m$ denotes the number of requests in the sequence. Later, the competitive ratio was improved to $O(m)$ by using a primal-dual deterministic algorithm \cite{Bienkowski:PrimalDualAlgo}.
\citet{AzarDeterministicMPMD} combined the ideas of these papers to provide a deterministic algorithm which is $O\left(m^{\log(2/3+\varepsilon)}\right)$-competitive for the $2$-MPMD and the MBPMD problems. 

Delaying decisions is also a well-known concept in the broader online domain. Already the classic ski rental problem \cite{KarlinKR2001,KarlinMMLO1990,KarlinMRS1986} postpones decisions to achieve better competitive bounds. Unlike the ski rental and similar rent-or-buy problems \cite{DoolyGS1998,DoolyGS2001}, matching is \emph{combinatorial} in nature, which complicates matters significantly. One related combinatorial problem that was considered with delays is the $k$-server problem. \citet{azar2017online} proposed a variant of the online $k$-server with delays and designed a so-called preemptive service algorithm, which achieves an $O(k\log^5 n)$ competitive ratio. Other related problems are the online bin-packing problem and the facility location problem, and the set cover problem which have also been considered in the scope of delays \cite{AzarBinPacking, azar2019general,azar_et_al:LIPIcs:2020:12874}.

\section{The $k$-MPMD Problem}

The goal of this section is to formally define the $k$-MPMD problem. We are a priori given a finite \emph{generalized metric space} $M = (V,d)$. Let $R$ be a sequence of requests in this metric space. Each request $\rho \in R$ is characterized by its location $\ell(\rho) \in V$ and an arrival time $t(\rho) \in \mathbb{R}^+$. 

The goal of the algorithm is to construct a (perfect) matching of the request set, namely, a partition of $R$ into $|R|/k$ request sets $S_i$, each of which contains $k$ requests. The assignment of requests to sets must be performed in an online fashion without withdrawal.  
The corresponding online algorithm $\A$ has to 
minimize (i) the matching cost and (ii) the incurred time delay.
For the sake of simplicity, we assume that the total number of requests is a multiple of $k$.

The matching cost for goal (i) is defined as follows: if $k$ requests $\rho_1, \mydots, \rho_k$ are matched, the algorithm needs to pay a \emph{space cost} $d(\Location(\rho_1), \mydots, \Location(\rho_k))$. The total space cost of algorithm $\A$ on the request set $R$ is then  $$\sCost_{\A}(R) = \sum_{\substack{S_i=\{\rho_{1,i}, \mydots, \rho_{k,i}\}\\ \  i\in\left[|R|/k\right]}} d(\Location(\rho_{1,i}), \mydots, \Location(\rho_{k,i})).$$ 

\noindent For the second goal we need to define a time delay: if algorithm $\A$ matches the request $\rho \in R$ at time $t'(\rho)$, then $\rho$ is said to be \emph{open} at all times $t(\rho) \le t < t'(\rho)$. For such open requests $\rho$ we need to pay a \emph{time cost} which is defined to be linear in the total waiting time $\left[t(\rho), t'(\rho)\right)$. The total time cost of algorithm $\A$ is then defined as $\tCost_{\A}(R) = \sum_{\rho \in R} \left(t'(\rho)-t(\rho)\right).$ 

\noindent The total cost incurred by algorithm $\A$ is defined as the sum of the space and the time cost:  $\Cost_{\A}(R) = \sCost_{\A}(R) + \tCost_{\A}(R).$
\noindent The goal of the online algorithm is then to assign requests $\rho\in R$ to sets $S_i$ such that the total cost $\Cost_{\A}(R)$ is minimized.

As we are interested in the competitive ratio of the online algorithm, we compare its performance to an optimal offline algorithm which knows the inputs sequence $R$ a priori. We denote this optimal offline algorithm $\OPT$. 
If $\A$ is a randomized algorithm, we define $\mathbb{E}[\Cost_{\A}(R)]$ as the expected cost that algorithm $\A$ incurs on the input sequence $R$. The competitive ratio of the randomized algorithm $\A$ is then defined as $\sup_R  \mathbb{E}[\Cost_{\A}(R)]/\Cost_{\OPT}(R).$
\noindent We say that a randomized algorithm is not competitive if its competitive ratio is unbounded. Our goal is therefore to find an algorithm $\A$ which has a 
small competitive ratio. 

\section{Generalized Metrics}\label{sec:generalized_metrics}

In order to study the $k$-MPMD problem, we need a metric that defines a distance between $k>2$ points. 
Several metric generalizations that formalize the notion of distance among three points have been proposed in the literature. The most prominent examples are the $2$-metrics \cite{Gahler1963}, the $D$-metrics \cite{Dmetric} and the $G$-metrics \cite{mustafa2006}. 

Metrics on more than three points were also considered in the literature. 
\cite{KamAlKhan2012} proposed a set of axioms on $k$ points which he called $K$-Metrics. This set is a generalization of the $G$-metrics. A less strict version of this set was later proposed by \cite{SamAsNMetric} and are called the $n$-Metrics. 

\begin{figure}
\centering
\begin{adjustbox}{max width=\textwidth}
\begin{tikzpicture}
   \draw[orange!70,fill=orange!7,rounded corners=1ex] (-0.15,0.15) rectangle (16.3,-4.95);

  \draw[black!40,fill=black!5,rounded corners=1ex] (2.6,0) -- (15.45,0) -- (15.45,-4.8) -- (0,-4.8) -- (0,-2.1) -- (2.6,-2.1) --cycle;
  
  \draw[red!60,fill=red!5,rounded corners=1ex, opacity=0.7] (0,-0.15) rectangle (4.9,-1.95);
  \draw[VioletRed!60,fill=VioletRed!5,rounded corners=1ex] (0.15,-0.3) rectangle (2.4,-1.8);
  
  \draw[Violet!60,fill=Violet!5,rounded corners=1ex]   (0.15,-2.25) rectangle (15.3,-4.65);
  \draw[blue!60,fill=blue!5,rounded corners=1ex]     (0.3,-2.4) rectangle (12.1,-4.5);
  \draw[NavyBlue!60,fill=NavyBlue!5,rounded corners=1ex] (0.45,-2.55) rectangle (9.05,-4.35);
  
  \node (labelleft) at (2.6,0.45)   {\Large{$3$-way}};
  \node (labelright) at (10,0.45)   {\Large{$k$-way}};  
  
  \draw[dash pattern=on 8pt off 2pt,line width=2pt] (5.1,0.7) to (5.1,-5) {};
  
  \node (2metric) at (1.135,-1.05)   {$\huge{\pmb{2}}$};
  \node (Dmetric) at (3.6,-1.05)  {$\huge{\pmb{D}}$};
  \node (Gmetric) at (2.67,-3.45)  {$\huge{\pmb{G}}$};
  \node (Hmetric) at (6.5,-3.45)  {$\huge{\pmb{H}}$};
  \node (Kmetric) at (10.5,-3.45)  {$\huge{\pmb{K}}$};
  \node (nmetric) at (13.6,-3.45)  {$\huge{\pmb{n}}$};
  
  \node[orange!90] (upperbox) at (15.9,-0.4)   {\Large{$\Pi$}};
  \node[black!80] (upperbox) at (14.8,-0.5)   {\Large{$O_D$}};
  
  \node[VioletRed!80] (upperbox) at (2,-1.5)   {\Large{$O_2$}};
  \node[red!80] (upperbox) at (4.45,-1.65)   {\Large{$\Delta_D$}};
  
  \node[NavyBlue!80] (upperbox) at (8.05,-4.05)   {\Large{$\Delta_H, \mathcal{S}_H$}};
  \node[blue!80] (upperbox) at (11.65,-4.2)   {\Large{$\mathcal{S}_K$}};
  \node[Violet!80] (upperbox) at (14.75,-4.35)   {\Large{$\Delta_n$}};
  
\end{tikzpicture}
\end{adjustbox}
\caption{Overview of all metrics. Each box represents one or more axioms (not metrics). Reading example: The $K$-metrics and the $n$-metrics both obey axioms \Per,\ \OD, and\ \ndelta; $K$-metrics additionally obey axiom\ \SK.}
\label{fig:metric_comparison}
\end{figure}
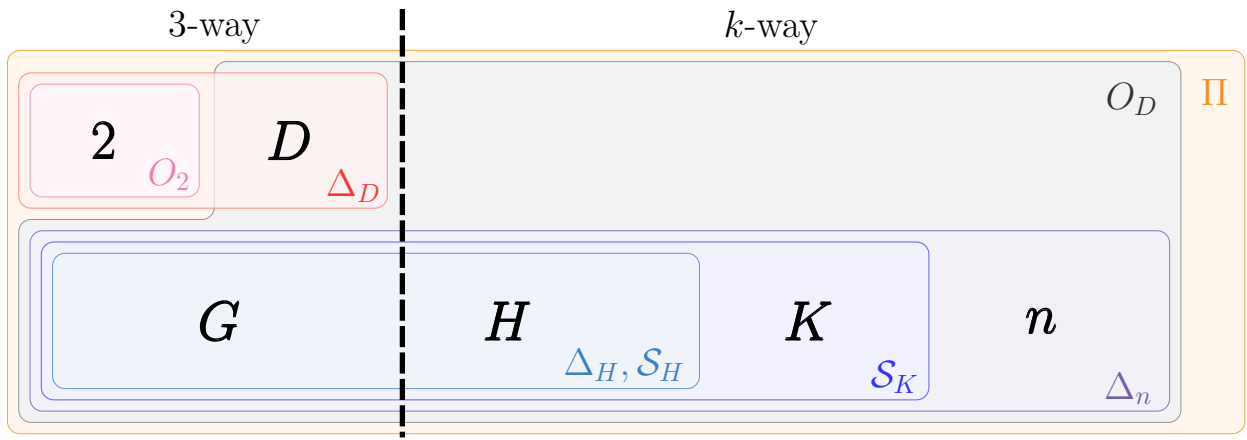

Figure \ref{fig:metric_comparison} visualizes the relation of these metrics with respect to the axioms that they satisfy. Note that all mentioned metrics satisfy the same generalized version of the symmetry axiom \Per. In addition, all metrics besides the $2$-metrics satisfy a generalized version of the positive definiteness \OD. These axioms are defined as follows:

\begin{axiomPi}
$d(v_1,\mydots,v_k)= d(\pi(v_1,\mydots,v_k))$, where $\pi(\cdot{})$ is a permutation of $\{v_1,...,v_k\}$
\end{axiomPi}

\begin{axiomOD}
$d(v_1,v_2,\mydots,v_k)=0 \Leftrightarrow$ $v_1=\mydots=v_k$, and $d(v_1,v_2,\mydots,v_k)>0$ otherwise.
\end{axiomOD}

\noindent The $2$- and the $D$-metrics both satisfy the triangle inequality \ddelta. We define this axiom on three points only:

\begin{axiomDeltaD}
$d(v_1,v_2,v_3) \leq d(v_1,v_2,a)+d(v_1,a,v_3) + d(a,v_2,v_3) \quad$
\end{axiomDeltaD}

\noindent Compared to the $D$-metrics, the $2$-metrics have to satisfy a different version of positive definiteness. Note that the $2$-metrics are not a subset of the $D$-metrics. 

\begin{axiomOTwo}
$d(v_1,v_2,v_3) = 0 \Leftrightarrow$ at least two of $v_1,v_2,v_3$ are equal, and $d(v_1,v_2,v_3)>0$ otherwise.
\end{axiomOTwo}

\noindent The $n$-metrics are the most general version of the $G$-metrics for $k$ points. Next to axioms \OD\ and \Per, it also has to satisfy the following version of the triangle inequality:

\begin{axiomDeltaN}
$d(v_1,\mydots,v_k)\leq d(v_1,\mydots,v_{k-1}, a) + d(\underbrace{a,\mydots, a}_{k-1},v_k) $ 
\end{axiomDeltaN}

\noindent In addition to the axioms of the $n$-metrics, the $K$-metrics make a separation between sets that contain exactly two and sets which contain exactly $k$ different elements. We denote this axiom as the separation axiom \SK. Note, that in literature, axiom \OD\ is sometimes also referred to as the separation axiom. The axiom \SK\ is defined as

\begin{axiomSK}
$d(\underbrace{v_1,\mydots,v_1}_{k-1}, v_2) \leq d(v_1,v_2,\mydots,v_k)$ where all elements $v_2, ..., v_k$  are distinct.
\end{axiomSK}

The $G$-metrics are a special case of the $K$- and the $n$-metrics, and they are equivalent to the $K$-metrics when $k=3$ is chosen. Observe that for some special examples, the $D$- and the $G$-metrics overlap. In Section \ref{sec:new_G_k_metric}, we will introduce the $H$-metrics, which are a generalization of the $G$-metrics. Compared to the $n$- and the $K$-metrics, the $H$-metrics have to satisfy a more strict version of the separation axiom and of the triangle inequality. These stricter versions of the axioms will naturally follow from the impossibility results presented in the next section.

\section{Impossibility Results}\label{sec:impossibility}
In this section, we will show that there exists no competitive algorithm for the $3$-MPMD problem on the metric spaces defined by the $2$- and the $D$-metrics. The only remaining metric space on $3$ points are thus the $G$-metrics for which there is a competitive algorithm. The algorithm presented in Section \ref{sec:algo} covers the $G$-metrics as a special case. For $k>4$, no known generalization of the $G$-metrics have a competitive algorithm for the $k$-MPMD problem. This will be shown in Section \ref{sec:impossibilityKNMetric} for the $n$- and the $K$-metrics.

\subsection{Impossibility Results for $2$- and $D$-Metrics}\label{sec:Impossibility2}
A typical example of a $2$-metric is the area of triangles. Even though the triangle area appears to be a reasonable measure for space costs, it is already insufficient for $3$-MPMD problems, as we are going to show next. Note that the following results also extend to the $D$-metrics.  

\begin{theorem}\label{thm:impossible_2metric}
There exists no randomized online algorithm that achieves a finite competitive ratio against an oblivious adversary for the $3$-MPMD problem on any $2$-metric with at least $3$ points.
\end{theorem}

We prove this result by applying Yao's minimax principle and by giving a distribution over inputs on which the cost of the optimal offline algorithm is almost $0$, while the expected cost of any deterministic online algorithm is not negligible. 
By repeating this input distribution sufficiently many times, the cost of the online deterministic algorithm can be arbitrarily large, but the cost of the optimal offline algorithm is still close to $0$. 

We first choose three different points from the $2$-metric space $(V,d_2)$, denoted by $v_1,v_2,v_3 \in V$. Without loss of generality, we assume that $d_2(v_1,v_2,v_3) = 1$. 
The construction of the bad example includes two request patterns, $P_1$ and $P_2$, each of which contains three requests. 
The three requests in pattern $P_1$ are all at point $v_1$, while pattern $P_2$ contains one request at every point $v_i, i=1,2,3$.

The distribution over request sequences depends on two parameters: the number of phases $r$ and a small time interval $0<\tau<1/3$. The time gap between any two consecutive phases is $1$. 
For each phase $i = 1,\mydots, r$, the following steps are executed: 
\begin{enumerate}
\item Present the three requests in $P_1$ to the algorithm simultaneously. 
\item Wait for time $\tau$.
\item Sample a random variable $C_i$ from a Bernoulli distribution, such that $P(C_i=1) = \frac{1}{r-i+1}$. 
\item If $C_i = 1$, then present requests in $P_2$ and terminate; Otherwise move on to the next phase. 
\end{enumerate}

This construction for the input sequence will serve as a baseline for impossibility constructions on generalized metrics, where we will only redefine the metric space, $P_1$ and $P_2$. In order to show that this construction leads to an unbounded competitive ratio for the $2$-metric, we will first show that the optimal offline algorithm always has a small cost for the defined request patterns $P_1$ and $P_2$. 

\begin{lemma}\label{lm:opt_cost}
The cost of $\OPT$ for every request set in the above distribution is $3\tau$. 
\end{lemma}
\begin{proof}
For phases in which only requests from $P_1$ arrive, $\OPT$ immediately matches the three requests without any cost. For the last phase in which six requests arrive, $\OPT$ waits for time $\tau$, matches the two requests $v_1$ with one request $v_2$, and matches the other three requests $v_1$, $v_1$, $v_3$. Due to axiom \OTwo, the cost of $\OPT$ equals the waiting for the three requests from $P_1$ in the last phase. 
\end{proof}

All possible deterministic online algorithms on these request sets can be represented by a vector $b$ of size $r$, $b = (b_i)_{i=1}^r \in \{0,1\}^r$. An entry $b_i = 1$ in this vector means that the algorithm waits for time $\tau$ in phase $i$. Then, if new requests from pattern $P_2$ arrive (i.e., in the termination phase), the deterministic algorithm matches two requests at $v_1$ and one request at $v_2$, and also matches the three requests at $v_1$, $v_1$ and $v_3$. Otherwise, the algorithm matches three requests at $v_1$. An entry $b_i=0$ means that the algorithm matches three requests in $P_1$ directly without incurring waiting cost. A request set may have less than $r$ phases, meaning that possibly not all elements in the vector are used. 

\begin{lemma}\label{lm:alg_bound}
For any deterministic online algorithm $\ALG$, the expected cost of the algorithm is at least $3/4$ if $\tau = 1/r$. 
\end{lemma}
\begin{proof}
Note that by construction a request set ends at every phase with the same probability of $1/r$. 

Let $B = \sum_{i=1}^r b_i.$ The probability that the algorithm guesses correctly (i.e., the request set terminates in phase $i$ and $b_i=1$) is $B/r$. 
Consider the event representing that the deterministic online algorithm guessed correctly. 
The expected waiting cost conditioned on this event is $3\tau (B+1) / 2$. 

The probability that the algorithm guesses incorrectly (i.e., the request set terminates in phase $i$ and $b_i=0$) is $1-B/r$. 
In this case, the cost of the algorithm is at least $1$, because the deterministic algorithm does not wait at the termination phase and has to match the three requests in $P_2$, which results in cost $1$. 

If we choose $\tau = 1/r$, 
the expected cost in total is at least $\frac{B}{r}\cdot \frac{3(B+1)\tau}{2} + 1-\frac{B}{r} = 1 - \frac{B}{r} + \frac{3B(B+1)}{2r^2} \ge 1 - \frac{B}{r} + \frac{B^2}{r^2}$. 
Note that the value $1 - \frac{B}{r} + \frac{B^2}{r^2}$ reaches its minimum for $\frac{B}{r} = 1/2$. That is, $1 - \frac{B}{r} + \frac{B^2}{r^2}\ge 3/4$. 
\end{proof}

Theorem \ref{thm:impossible_2metric} follows by combining Lemma \ref{lm:opt_cost}, Lemma \ref{lm:alg_bound} and Yao's minimax principle\footnote{In fact, we can repeat the above process (infinitely) many times to show that Theorem \ref{thm:impossible_2metric} holds even if an additive term is allowed in the competitive analysis.}. 
For the $D$-metric, we can prove an analogous statement:

\begin{theorem}\label{thm:impossible_Dmetric}
There exists no randomized algorithm for the $3$-MPMD problem on $D$-metrics against an oblivious adversary that has a competitive ratio which is bounded by a function of the number of points $n$. 
\end{theorem}

Note that Theorem \ref{thm:impossible_2metric} stated that for any non-trivial $2$-metric it is impossible to design a competitive algorithm. In contrast, Theorem \ref{thm:impossible_Dmetric} does not exclude the possibility to design competitive algorithms for some specific non-trivial $D$-metrics.  
The detailed proof of Theorem \ref{thm:impossible_Dmetric} is omitted, since it is similar to the proof of Theorem \ref{thm:impossible_2metric}. As axiom \OD\ is different from axiom \OTwo, instead of having $d(v_1,v_1,v_2) = d(v_1,v_1,v_3) = 0$, we can let $d(v_1,v_1,v_2)$ and $d(v_1,v_1,v_3)$ be arbitrarily close to $0$. By setting $d(v_1,v_2,v_3)=1$, we can prove analogous statements to Lemma \ref{lm:opt_cost} and \ref{lm:alg_bound} in order to show that the competitive ratio of any randomized algorithm will be unbounded. 
These results show that one cannot use a $2$-metric or $D$-metric to model matching costs. 
In fact, from these impossibility results and their proofs, one can conclude that, to some extent, the $G$-metrics are necessary for $3$-MPMD problems.

\subsection{Impossibility Result for $n$- and $K$-Metrics}\label{sec:impossibilityKNMetric}
Other than the $2$- and $D$-metrics, the $n$- and $K$-metrics in this section can be shown to give positive results for the case $k=3$ and $k=4$ respectively. However, for $k > 3$ (resp. $k>4$), it can be shown that the competitive ratio of any online algorithm is unbounded. 

\begin{theorem}\label{thm:impossible_nmetric}
There exists no randomized algorithm for the $k$-MPMD ($k\geq 4$) problem on $n$-metrics against an oblivious adversary that has a competitive ratio which is bounded by a function of the number of points $n$. 
\end{theorem}
\begin{proof}
The idea for this proof is similar to the one in the previous section. We will first define an example of an $n$-metric as follows:
Given a set $V$ with $k+1$ elements $v_1,\mydots,v_{k+1}$, we define distances for any subset of elements as
\begin{itemize}
    \item $d(v,\mydots,v) = 0\quad \forall v\in V$ 
    \item $d(v_1,v_2,\mydots,v_k) = \varepsilon$ if all $v_1,\mydots,v_k$ are distinct
    \item $d(v_1,v_2,\mydots,v_k) = 1$ otherwise
\end{itemize}
The proposed distances indeed define an $n$-metric, as axioms \Per\ and\ \OD\ follow directly from the definition of the distances. In order to show that this metric satisfies the triangle inequality \ndelta, we first assume that the set on the left-hand side consists of only one element, then its distance is $0$ and \ndelta\ is satisfied trivially. If the set contains $k$ different elements, the right-hand side will contain at least one set with distance $\geq \varepsilon$, as it cannot be the sum of two sets which both contain exactly one element. If the left-hand side contains $x\in[2,k-1]$ distinct elements, then either $d(\underbrace{a,\mydots, a}_{k-1},v_k)$ has exactly two different elements and thus distance $1$, or $v_k = a$ in axiom \ndelta\ and thus the set $\{v_1,\mydots,v_{k-1}, a\}$ on the right-hand side is equal to the set  $\{v_1,\mydots,v_{k-1}, v_k\}$ on the left-hand side. 

For the worst-case input sequence, let $P_1 = \{v_1,\mydots,v_k\}$ and $P_2 := \{v_3,\mydots,v_{k+1},v_{k+1}\}$ be the two request patterns, where $v_1,\mydots,v_{k+1}$ are $k+1$ distinct elements. Assuming that $P_1$ and $P_2$ arrive in the same fashion as described in Section \ref{sec:Impossibility2}, the offline algorithm will serve every pattern $P_1$ except the last one at cost $\varepsilon$. The last pattern $P_1$ and pattern $P_2$ will be served together, by reordering the request sets to $\{v_2,\mydots,v_{k+1}\}$ and $\{v_1,v_3,\mydots,v_{k+1}\}$. In this case, the cost of $\OPT$ will be $\varepsilon$ for every arriving request plus additional waiting cost $k\cdot\tau$ for the last phase.
\end{proof}

This counter example only works because axiom \SK\ of the $G$- and the $K$-metrics does not have an equivalent counterpart in the $n$-metrics. Indeed, the presented counter example does already not satisfy axiom \SK, as $d(v_1,v_2,\mydots,v_k) = \varepsilon$ for distinct elements $v_i$ and $d(a,\mydots,a,v_1)=1$ for $a\neq v_1$. While \SK\ of the $K$-metrics defines a separation between sets with $2$ and sets with $k$ distinct elements, such a separation is not defined for any other pair of distinct sets (e.g., sets with $4$ and sets with $5$ distinct elements). Therefore, it is possible to adapt the above counter example to derive a similar impossibility result for axiom \SK\ as well. However, the axiom \SK\ is not the only axiom which fails in this generalization. The following theorem states that the $K$-metrics fail to be competitive using the triangle inequality axiom \ndelta, the corresponding proof is given in Appendix~\ref{app:impossibility_proofs}.

\begin{theorem}\label{thm:impossible_Kmetric}
There exists no randomized algorithm for the $k$-MPMD ($k\ge 5$) problem on $K$-metrics against an oblivious adversary that has a competitive ratio which is bounded by a function of the number of points $n$. 
\end{theorem}

Theorem \ref{thm:impossible_nmetric} and \ref{thm:impossible_Kmetric} show that the existing generalizations of the $G$-metrics fail to be competitive for $k > 3$ and $k > 4$ respectively. The presented counter examples also suggest that generalized $G$-metrics, which have a bounded competitive ratio, have to contain reasonable generalized versions of the separation axiom and the triangle inequality.

\section{The New Generalized Metric Space $H$}\label{sec:new_G_k_metric}

In this section we will present a more restricted generalization of the $G$-metrics to $k$ points. The restriction of the triangle inequality \ndelta\ follows directly from the counter example of Theorem \ref{thm:impossible_Kmetric}:

\begin{axiomDeltaH}
$d(v_1,\mydots,v_k)\leq d(v_1,\mydots, v_i, \underbrace{a,\mydots,a}_{k-i}) + d(\underbrace{a,\mydots,a}_{i},v_{i+1},\mydots,v_k)\ \forall\  a\in V$  and\ \ $\forall\ i \in [k]$
\end{axiomDeltaH}
Theorem \ref{thm:impossible_nmetric} does not directly imply the required version of axiom \SH. As mentioned before, the counter example can be also extended to violate axiom \SK\ of the $K$-metrics. Following such an extended counterexample, axiom \SK\ can be strengthened as follows:

\begin{axiomSHstar}
$d(S_i)\leq d(S_j),\ \text{if}\ elem(S_i)\subset elem(S_j)$ 
\end{axiomSHstar}
\noindent where $elem(S)$ denotes all distinct elements contained in multiset $S$.
However, the presented generalizations are not yet sufficient to be able to find a competitive algorithm. Note that we can extend the counter example of Theorem \ref{thm:impossible_nmetric} to also consider sets which contain the same elements. The next theorem states the corresponding impossibility result, the corresponding proof can be found in Appendix~\ref{app:impossibility_proofs}. 

\begin{theorem}
There exists no randomized algorithm for the $k$-MPMD ($k\ge 5$) problem on the $K$-metrics enriched with the axioms \SHstar\ and \hdelta\ against an oblivious adversary that has a competitive ratio which is bounded by a function of the number of points $n$.
\end{theorem}

Following the above impossibility result, we can extend the \SHstar\ axiom to the desired \SH\ axiom. We call a metric which satisfies the following four axioms an $H$-metric.

\setcounter{axiomGk}{0}
\begin{description}
\item [$H$-Metric:]  The function $d:V^k\rightarrow[0,\infty)$ is called a $H$-\textit{metric} on a set $V$ if and only if for all $v_1,\mydots,v_k\in V$ the following axioms are satisfied:
\begin{axiomPi}
$d(v_1,\mydots,v_k) = d(\pi(v_1,\mydots,v_k))$, where $\pi(\cdot{})$ is a permutation of $\{v_1,...,v_k\}$
\end{axiomPi}
\begin{axiomOD}
$d(v_1,\mydots,v_k)=0 \Leftrightarrow$ $v_1=v_2=\mydots=v_k$, otherwise  $d(\underbrace{v_1,\mydots,v_1}_{k-1}, v_2)>0 \quad\forall\  v_1 \neq v_2$ 
\end{axiomOD}
\begin{axiomSH}
$d(S_i)\leq d(S_j),$ if $elem(S_i)\subset elem(S_j)$ and $d(S_i)\leq (k-1)\cdot d(S_j),$ if $elem(S_i) = elem(S_j)$
\end{axiomSH}
\begin{axiomDeltaH}
$d(v_1,\mydots,v_k)\leq d(v_1,\mydots, v_i, \underbrace{a,\mydots,a}_{k-i}) +d(\underbrace{a,\mydots,a}_{i},v_{i+1},\mydots,v_k)\ \forall\  a\in V$ and\ \ $\forall\ i \in [k]$
\end{axiomDeltaH}
\end{description}

The $H$ metric space contains a rich class of non-trivial functions. In the following, we will present an example of a metric underlying the presented set of axioms. Two more examples can be found in Appendix~\ref{sec:more_examples}.

\begin{example}
In this example, we consider a game with four players. The players appear on one of the three platforms $V=\{a,b,c\}$. Our metric space is defined on $V^4$. 

Assume that the platforms $a$ and $b$ as well as $a$ and $c$ are at distance $1$ to each other, while $b$ and $c$ are at distance $1/2$ from each other. In order to find a good match, we want to minimize the distances between the used platforms. We therefore define a penalty for any matching which corresponds to the length of the shortest path between all platforms that participate in the matching. In addition, we want to match players of the same platform whenever possible. This idea is represented by the \SH\ axiom.
We therefore will introduce an additional penalty for players who were matched across different platforms: a matching that matches exactly one player from one or more platforms will be assigned a penalty of $1$. That is, in the matching $(a,b,b,c)$, platforms $a$ and $c$ will contribute to a total penalty of $1$.
The total weight $d$ of a matching is determined by the sum of the penalties for the distance between platforms and the penalties for little used platforms.

Following these rules, we can define the weights of the matchings in this metric as
\begin{itemize}
\item $d(a,a,a,a) = d(b,b,b,b) = d(c,c,c,c) = 0 + 0$
\item $d(b,c,c,c)= d(c,b,b,b)= \frac{1}{2} + 1$
\item $d(a,b,b,b) = d(a,c,c,c) = d(b,a,a,a)= d(c,a,a,a) = 1+1$ 
\item $d(a,a,b,b) = d(a,a,c,c) = 1 + 0$ 
\item $d(b,b,c,c) = \frac{1}{2} + 0 $
\item $d(a,a,b,c) = d(a,b,b,c) = d(a,b,c,c) = 1.5 + 1$
\end{itemize}
Note that all other weights $d$ can be extended to $V^4$ using the property of symmetry (\Per). In this example, we can verify that matchings that have single players from some platform also have a higher weight. We next will verify the axioms of an $H$-metric. The first two axioms \Per\ and \OD\ hold by definition. The axiom \SH\ holds because matchings that contain three different platforms all have weight $2.5$ which is larger than any matching over exactly two platforms; matchings that contain two platforms have a non-zero weight, that is, the weight is larger than the weight of matchings over only one platform; in addition, we need to verify that the matchings across the same two platforms differ by at most a factor of $(k-1)=3$ in weight. We will omit verifying each inequality here. Observe, however, that for $d(b,b,c,c)$ and $d(c,b,b,b)$ the inequality $d(c,b,b,b) \leq 3\cdot d(b,b,c,c)$ is tight.
It remains to check the triangle inequality \hdelta. The interesting inequalities are ones where the number of matched platforms is different on the right and the left hand-sides of the inequality:
\begin{itemize}
\item $d(a,b,b,b) \leq d(b,b,b,b) + d(a,b,b,b) = 0 + 2 $
\item $d(a,b,c,c) \leq d(b,b,c,c) + d(a,b,b,b) = \frac{1}{2} + 2$
\item $d(a,a,b,c) \leq d(a,a,c,c) + d(c,c,b,c) = 1+ 1.5$
\end{itemize}
Other inequalities can be verified accordingly. With this, we have verified that our example indeed is an $H$-metric. Finally, observe that we would have also received an $H$-metric without penalizing the little used platforms. We chose this example to show that the \SH\ and the \hdelta\ axioms can both be satisfied with equality in some metric spaces.
\end{example}

\subsection{Reduction}\label{sec:reduction}
In this section, we make the observation that the presented $H$-metrics are roughly equivalent to a metric defined through pairwise distances of the $k$ points:

\begin{theorem}\label{thm:reduction}
Let $d_H$ be an $H$-metric on a set $V$. 
There exists a metric $d:V^2\rightarrow[0,\infty)$ such that 
\begin{eqnarray*}
c_l\cdot \sum\limits_{i=1}^{k-1}\sum\limits_{j=i+1}^{k} d(v_i,v_j) \le d_H(v_1,v_2,\mydots,v_k)\\ 
\le c_u\cdot \sum\limits_{i=1}^{k-1}\sum\limits_{j=i+1}^{k} d(v_i,v_j)
\end{eqnarray*}
for two constants $c_l$ and $c_u$ which are only dependent on $k$.
\end{theorem}

We will prove this theorem in Appendix~\ref{sec:proof of the reduction thm}. For the following section, we need to define the desired metric $d$ on pairs of points with respect to the $H$-metrics: 
\begin{definition}\label{definition:metric}
For $v_1, v_2 \in V$ we define a metric $d:V^2\rightarrow[0,\infty)$ as
$$d(v_1,v_2) := d_H(v_1,v_2,\mydots,v_2) + d_H(v_2,v_1,\mydots,v_1)$$
\end{definition}

\renewcommand*{\proofname}{Proof}

\section{An Algorithm for the $k$-MPMD Problem on $H$-Metrics}\label{sec:matching_alg}

In this section, we present an algorithm for the $k$-MPMD problem. The corresponding correctness analysis will be presented in Section \ref{sec:correctness_analysis}.
This algorithm follows the common idea used in almost all existing MPMD algorithms: 
embed the given $H$-metric $\mathcal{M}$ into tree metrics, reduce the height of the tree
and, finally, design an online algorithm on the tree. Let $(T,w)$ denote a tree metric where $T$ is the tree and $w$ is the corresponding metric defined on this tree.
The first steps of the algorithm - to embed the $H$-metric $\mathcal{M}$ to a tree metric $(T,w)$ and reduce the height - are done according to Section \ref{sec:reduction} and the results in \cite{Fakcharoenphol2004485} and \cite{Bansal:2015}:
\begin{enumerate}
    \item Define metric $d$ according to Definition  \ref{definition:metric}.
    \item Apply the result in \cite{Fakcharoenphol2004485} to embed metric $d$ obtained in the previous step to a tree metric.
    \item Reduce the height of the tree obtained in the previous step by applying the result in \cite{Bansal:2015}. The resulting tree metric is denoted by $(T,w)$.

\end{enumerate}
By combining the three results, any finite $H$-metric $\M=(V,d)$ can be efficiently embedded into a tree metric $(T,w)$ 
with height $O(\log n)$ and an expected distortion of $O(\log n)$. The following lemma formally describes this idea.

\begin{lemma}\label{lemma:treeEmbed}
For any fixed $k$, every $H$-metric $(V,d_H)$ can be probabilistically embedded into a tree metric $(T,w)$ of height $O(\log n)$ such that for all $(v_1, v_2,\mydots,v_k )\in V^k$ holds: 
\begin{itemize}
    \item $d_H(v_1,v_2,\mydots,v_k) \le \sum\limits_{i=1}^{k-1}\sum\limits_{j=i+1}^{k} w(v_i,v_j),$ and 
    \item $\sum\limits_{i=1}^{k-1}\sum\limits_{j=i+1}^{k} \mathbb{E}\left[w(v_i,v_j)\right] \le O(\log n)\cdot d_H(v_1,v_2,\mydots,v_k).$
\end{itemize}
\end{lemma}

\subsection{An Algorithm on Tree Metrics}\label{sec:algo}

We design an online algorithm $\A$ on any tree metric $(T,w)$ where $T$ 
is a tree rooted at node $r$, $w$ a weight function on edges of $T$ and $\Leaves$ the set of all leaves in $T$. 
We assume that the tree used in this subsection is obtained by applying Lemma \ref{lemma:treeEmbed}. 
For any vertex $v$ in the tree, let $T_v$ denote the subtree of $T$ rooted at $v$, $\Leaves_v$ denote the leaves of $T_v$, 
$e_v$ denote the edge between $v$ and its parent, and $w_v$ denote the weight of $e_v$. Note
that the weight of $e_v$ corresponds to the distance between the two ends of this edge. For simplicity, we can also define an edge $e_r$ for the root which has the weight $w_r \coloneqq \infty$.
Let $\Ancestors(v)$ and $\Descendant(v)$ 
be the sets of ancestors and descendants of $v$ respectively. 

We assume that the requests arrive only at the leaves $\Leaves$ of the tree. This assumption is reasonable, since a request on a non-leaf node can be instead considered as a request on a leaf node that has distance $0$ to the non-leaf node. 
Given some vertex $v$ in $T$, we denote the set of open requests in $\Leaves_v$ at time $t$ by $\Active_{v}(t)$. 
Open requests under the offline algorithm are denoted by $\adv{\Active}_v(t)$.
Note that we can assume that each leaf in $\Leaves$ hosts at most $k-1$ open requests, i.e., $|\Active_v(t)|\le k-1$ for all $v\in \Leaves$. Consider therefore some fixed point of time $t$. 
If request $\rho$ arrives at time $t$ at leaf node $\Location(\rho)$, and $\Location(\rho)$ already
hosts $k-1$ open requests (under $\A$), then an algorithm can match 
them immediately.

There exist different notions under which $k$ requests are matched with respect to the nodes inside the tree: consider a vertex $v$ in $T$ and suppose $k$ requests $\rho_1, \mydots, \rho_k$ are matched. 
If $v$ is the least common ancestor~(lca) of $\Location(\rho_i), i=1,\mydots,k$ 
then we say that the requests are \emph{matched across} $v$. 
If $v \neq \LCA(\Location(\rho_1), \mydots, \Location(\rho_k))$, but $v$ is on the shortest path determined by $\Location(\rho_i),\Location(\rho_j)$ for some $i,j \in [k]$, 
then we say that the requests are \emph{matched on top of} $v$. 
If $v$ is an ancestor of $\LCA(\Location(\rho_1), \mydots, \Location(\rho_k))$, we say that the requests are \emph{matched under} $v$. 

The idea of the algorithm is to match requests that are close to each other as soon as possible, by letting nodes on the lower level of the tree have an advantage when matching requests. On the other hand, requests that have waited too long to be matched in their neighborhood, should be able to get matched with nodes at a farther distance in order to also minimize the total waiting time.

\begin{algorithm}[tbh]
\caption{A Deterministic Algorithm for $k$-MPMD on Tree Metrics
}\label{algo}
\begin{flushleft}
\textbf{Initialization:}
Let every timer be active and the initial value be $0$.

\textbf{At every moment:}
\begin{itemize}
    \item While there are $k$ unmatched requests at the same point, match those requests immediately. 
    \item If there exist $k$ requests such that the edges on the shortest path connecting them are all inactive, then match these requests, and let the timers on these edges become active. 
    \item For each vertex $u$, if the corresponding timer is active and $|C_u| \nequiv 0 \mod k$,  then increase $\tau_u$ at the unit rate. Else, pause the timer. 
    \item For each vertex $u$, as soon as the value of the corresponding timer $\tau_u$ becomes equal to some integral multiple of $w_u$, then we let the timer be inactive and pause the timer. 
\end{itemize}
\end{flushleft}
\end{algorithm}

Algorithm \ref{algo} presents a solution for $k$-MPMD on tree metrics. It thereby associates every edge with a timer $\Timer_{v}\in \Reals_{\ge 0}$ which is initially set to $0$. 
For each node $v$, this timer increases at a unit rate if the set of open requests $\Active_{v}(t) \not \equiv 0 \mod{k}$, until it reaches an integral multiple of $w_v$. 
As soon as $\Timer_{v}$ reaches the next integral multiple of $w_v$, we pause the corresponding timer $\Timer_{v}$ and call it \emph{inactive}. An inactive timer is not allowed to continue running until its value is \textit{consumed} and it becomes active again. The collected time between two inactive periods can be consumed in a matching on top of $v$, while the actual value of the timer remains unchanged. 
For any $u\in T-\Leaves$, 
we match requests $\rho_1,\mydots,\rho_k$ across $u$ if and only if for all $i=1,\mydots, k$, every timer on the path connecting $u$ and $\Location(\rho_i)$ is inactive. 
After the matching, these inactive timers are consumed and become active again. Figure~\ref{fig:algorithm_example} depicts how active timers change when new requests arrive in the algorithm.

The following theorem states the correctness of the analysis and the competitive ratio of the algorithm $\A$, the corresponding proof will be presented in Appendix~\ref{sec:correctness_analysis}.

\begin{theorem}\label{thm:correctnessTheorem}
For any given $k$ and any request sequence $R$, algorithm $\A$ achieves a competitive ratio $O(\log n)$, i.e.,
$\Expect[\Cost_{\A}(R)] \le O(\log n)\cdot \Cost_{\OPT}(R). $
\end{theorem}

\section{Discussion}
In this paper, we focused on deriving a generalized metric defined on $k$ points which is suitable for $k$-way matching problems. Using Yao's minimax principle, we derived various counterexamples which have put restrictions on the desired generalized metric. Our results suggest that the $H$-metric is possibly the only metric that can be used to solve the $k$-MPMD problem. We further showed that the $H$-metric can be reduced to metrics where distances of $k$ points are defined through pairwise distances of the points in the considered set. This reduction gave us the possibility to embed our metric into a tree and modify known algorithms for the $2$-MPMD problem in order to solve the $k$-MPMD problem without extra cost.

While our presented results suggest that there is no fundamentally different generalized metric that can be applied to solve the $k$-MPMD problem, we believe that the presented metric can be of interest in many domains. On one hand, we only focused on the generalization of the $2$-MPMD problem which is only one version of a broader class of online matching problems, some of which were listed in Section~\ref{sec:relatedwork}. It would be interesting to investigate how other online matching problems can be generalized to $k$-way matching and whether the $H$-metric is the only suitable metric for such problems (if metrics are needed) as well, e.g.,~\cite{gupta2012online,nayyar2017input,raghvendra2016robust}. On the other hand, not only online matching problems require an underlying metric space, and it might be interesting to generalize other problems than online matching, e.g., offline matching~\cite{agarwal2014approximation}, online $k$-server~\cite{lee2018fusible}. 

While we presented our work in the juicy context of multi-player games, various other applications that need to group data that arrives online may benefit from our $k$-MPMD algorithm. For such applications, we would like to improve our algorithm and the competitive analysis by minimizing the effect of the parameter $k$. 

Note that we consider (generalized) metric spaces, i.e., spaces that satisfy a certain set of axioms. But, there is also a possibility to relax the coefficients in the (generalized) triangle inequality, which would not change the results much, or to use a completely different set of axioms, but that is for future work and this direction even has not been considered for $2$-MPMD, i.e., we wonder whether we can relax or tighten the condition of the standard triangle inequality in $2$-MPMD and design different competitive algorithms. 

\section*{Acknowledgments}
We would like to thank Tim Bohren and Kyriakis Panagiotis for their valuable input on the $H$-metric. We would also like to thank anonymous reviewers for their helpful comments and feedback on previous versions of this paper.

\bibliographystyle{plainnat}
\bibliography{literature}

\clearpage
\appendix

\setcounter{theorem}{3} 
\section{Impossibility Proofs}\label{app:impossibility_proofs}
\begin{theorem}
There exists no randomized algorithm for the $k$-MPMD ($k\ge 5$) problem on $K$-metrics against an oblivious adversary that has a competitive ratio which is bounded by a function of the number of points $n$. 
\end{theorem}
\begin{proof}
The proof once again starts by an example for the $K$-metric:
Given a set $V$ with $3$ elements $a,b,$ and $c$, we define distances for any subset of elements as follows:
\begin{itemize}
    \item $d(v,\mydots,v) = 0\quad \forall\ v\in V$
    \item $d(\underbrace{a,\mydots,a}_{\lfloor k/2\rfloor},\underbrace{b,\mydots,b}_{\lfloor k/2\rfloor}, \underbrace{b}_{k-2\lfloor k/2\rfloor}) = \varepsilon$
    \item $d(\underbrace{b,\mydots,b}_{\lfloor k/2\rfloor},\underbrace{c,\mydots,c}_{\lfloor k/2\rfloor}, \underbrace{b}_{k-2\lfloor k/2\rfloor}) = \varepsilon$  
    \item $d(v_1,v_2,\mydots,v_k) = 1$ otherwise
\end{itemize}  
In order to show that this is a $K$-metric, first note that axioms \Per,\ \OD\ and \SK\ are satisfied trivially by the definition. The only interesting axiom is the triangle inequality \ndelta. Assume that not all elements are the same on the left-hand side, otherwise the inequality is satisfied trivially. Then, independent of the set on the left-hand side, the right-hand side will always have distance of at least $1$. This is because only three possible sets have distance less than $1$ and, for $k\geq4$, no two of them can represent the sets on the right-hand side of the inequality simultaneously. 

We can now define the request patterns as $P_1 = \{\underbrace{b,\mydots,b}_{k}\}$ and \\$P_2 := \{\underbrace{a,\mydots,a}_{\lfloor k/2\rfloor},\underbrace{c,\mydots,c}_{\lfloor k/2\rfloor}, \underbrace{b}_{k-2\lfloor k/2\rfloor}\}$. Matching request in pattern $P_1$ always has cost $0$. Matching request in $P_2$ has always cost $1$. If requests from both patterns $P_1$ and $P_2$ are matched at the same time, the values in the sets can be rearranged to $\{\underbrace{a,\mydots,a}_{\lfloor k/2\rfloor},\underbrace{b,\mydots,b}_{\lfloor k/2\rfloor},\underbrace{b}_{k-2\lfloor k/2\rfloor}\}$ and $\{\underbrace{b,\mydots,b}_{\lfloor k/2\rfloor},\underbrace{c,\mydots,c}_{\lfloor k/2\rfloor},\underbrace{b}_{k-2\lfloor k/2\rfloor}\}$, such that the matching cost is $2\varepsilon$. With the arrival strategy of $P_1$ and $P_2$ chosen as in Section \ref{sec:Impossibility2}, this will result in an unbounded competitive ratio for any randomized online algorithm.
\end{proof}

\begin{theorem}
There exists no randomized algorithm for the $k$-MPMD ($k\ge 5$) problem on the $K$-metric enriched with the axioms \SHstar\ and \hdelta\ against an oblivious adversary that has a competitive ratio which is bounded by a function of the number of points $n$.
\end{theorem}
\begin{proof}
 Given a set $V$ with $2$ elements $a$ and $b$, we define distances for any subset of $k$ elements as
\begin{itemize}
    \item $d(v,\mydots,v) = 0\quad \forall v\in V$ 
    \item $d(\underbrace{v_1,\mydots,v_1}_{i},\underbrace{v_2,\mydots,v_2}_{k-i}) = \varepsilon$ for $1<i<\lfloor(k-1)/2\rfloor$, where $v_1\neq v_2$
    \item $d(v_1,v_2,\mydots,v_k) = 1$ otherwise
\end{itemize}
This example satisfies the given metric properties: axioms \Per,\ \OD\ and \SHstar\ are satisfied trivially. The triangle inequality \hdelta\ is satisfied, because only one of the terms on the right-hand side (unless one of the terms is equal to the left-hand side) can be of distance $\varepsilon$. Further note that this example is not a contradiction to the $G$-metrics for $k=3$, since on three points, there would not exist sets of distance $\varepsilon$.

We can define the request patterns as $P_1 \coloneqq \{b,\mydots,b\}$ and $P_2\coloneqq \{\underbrace{a,\mydots,a}_{2\cdot i},\underbrace{b,\mydots,b}_{k-2\cdot i}\}$. While matching requests in $P_1$ has cost $0$ and those in $P_2$ has cost $1$, combined, requests in both patterns can be reordered to two sets $\{\underbrace{a,\mydots,a}_{i},\underbrace{b,\mydots,b}_{k-i}\}$, that have cost $2\varepsilon$. Rest of the analysis is as in Section \ref{sec:Impossibility2}.
\end{proof}

\section{Examples of the $H$-metric}\label{sec:more_examples}
\begin{example}\label{simple_example}
Let $V = \{a,b\}$. The $H$-metric $d$ is defined on $V^4$ as follows 
\begin{eqnarray*}
d(a,a,a,a) = d(b,b,b,b) =0 \\
d(a,a,a,b) = d(a,a,b,b) =1 \\ 
d(a,b,b,b) = 2
\end{eqnarray*}
and let all other $d$ be extended to $V^4$ using the property of symmetry (\Per). Then it can be easily verified that $(V,d)$ is an $H$-metric space. 
\end{example}

The above example shows that $H$-metric is powerful as it allows $d(S_i) \neq d(S_j)$ even if $elem(S_i) = elem(S_j)$. The next example is more general and may capture a setting of many applications. 

\begin{example}
Let $V$ be a finite\footnote{The finiteness is not necessary to define the metric space, but it is needed for the online matching problem.} subset of $\mathbb{R}^\ell$. For any $k$-tuple $(v_1, v_2, \mydots, v_k) \in V^k$, we define $d(v_1, v_2, \mydots, v_k) := \max\{||v_i - v_j ||\mid 1\le i,j\le k\}$. It can be verified that $(V,d)$ is an $H$-metric space. 
\end{example}

\section{Proof of Theorem~\ref{thm:reduction}}\label{sec:proof of the reduction thm}
In order to prove Theorem~\ref{thm:reduction}, we first need to define the desired metric $d$ on pairs of points with respect to the $H$-metric:
\setcounter{definition}{0} 
\begin{definition}
For $v_1, v_2 \in V$ we define a metric $d:V^2\rightarrow[0,\infty)$ as
$$
d(v_1,v_2) := d_H(v_1,v_2,\mydots,v_2) + d_H(v_2,v_1,\mydots,v_1)
$$
\end{definition}
It can be easily verified that $d$ defined as above is indeed a metric defined on $k=2$ points. Using this definition, we can prove the previous theorem:
\renewcommand*{\proofname}{Proof of Theorem \ref{thm:reduction}}
\begin{proof}
We start by showing the right inequality by repeatedly applying axiom \hdelta: 
\allowdisplaybreaks
\begin{alignat*}{2}
 d_H(v_1,\mydots,v_k) & \overset{\Delta_H}{\leq} d_H(v_1,\mydots, v_{k-1}, v_1) + d_H(v_1,\mydots, v_1, v_k) \\
  & \overset{\Delta_H}{\leq} d_H(v_1,\mydots, v_{k-2},v_1, v_1) + d_H(v_1,\mydots, v_1,v_{k-1}, v_1)+ d_H(v_1,\mydots, v_1, v_k) \\
  & \overset{\Delta_H}{\leq} \mydots \overset{\Delta_H}{\leq} \sum\limits_{i=2}^{k} d_H(v_1,\mydots, v_1, v_i) \leq \sum\limits_{i=2}^{k} \left(d_H(v_1,\mydots, v_1, v_i) + d _H(v_i,\mydots, v_i, v_1)\right) \\
 & \leq \sum\limits_{i=2}^{k} d(v_1, v_i) \leq \sum\limits_{i=1}^{k-1}\sum\limits_{j=i+1}^{k} d(v_i,v_j)   
\end{alignat*}

For the left inequality we can use axiom \SH:
\begin{alignat}{2}\label{mylabel}
\sum\limits_{i=1}^{k-1}\sum\limits_{j=i+1}^{k} d(v_i,v_j) & =  \sum\limits_{i=1}^{k-1}\sum\limits_{j=i+1}^{k} \left( d_H(v_i,\mydots, v_i, v_j) + d_H(v_j,\mydots, v_j, v_i) \right)\\
 & \stackrel{(\ref{mylabel})}{\leq} \frac{k(k-1)}{2}\cdot2\cdot k\cdot d_H(v_1,\mydots,v_k) \leq k^3\cdot d_H(v_1,\mydots,v_k) 
\end{alignat}
The inequality (\ref{mylabel}) follows from the axiom \SH\ and a case distinction. If $\{v_1,\mydots,v_k\}$ contains at least three different elements, then each summand $d_H(v_i,\mydots, v_i, v_j)$ can be directly bounded by $\{v_1,\mydots,v_k\}$, since $v_i,v_j\in \{v_1,\mydots,v_k\}$ and $\{v_1,\mydots,v_k\}$ has more than the two elements $v_i$ and $v_j$. On the other hand, if $\{v_1,\mydots,v_k\}$ contains exactly two elements, then the sum\\ $\sum_{i=1}^{k-1}\sum_{j=i+1}^{k} \left( d_H(v_i,\mydots, v_i, v_j) + d_H(v_j,\mydots, v_j, v_i) \right)$ contains at most $\frac{k(k-1)}{2}<k^2/2$ summands, where each summand is bounded by $k\cdot d_H(v_1,\mydots,v_k)$.

Combining both directions results in 
$$\frac{1}{k^3}\cdot \sum\limits_{i=1}^{k-1}\sum\limits_{j=i+1}^{k} d(v_i,v_j) \leq d_H(v_1,\mydots,v_k) \leq \sum\limits_{i=1}^{k-1}\sum\limits_{j=i+1}^{k} d(v_i,v_j).$$
\end{proof}

\renewcommand*{\proofname}{Proof}

\section{Competitiveness Analysis}\label{sec:correctness_analysis}

\newcommand\ourclock{\clock[0.15]{0}{15}{0,0}}

 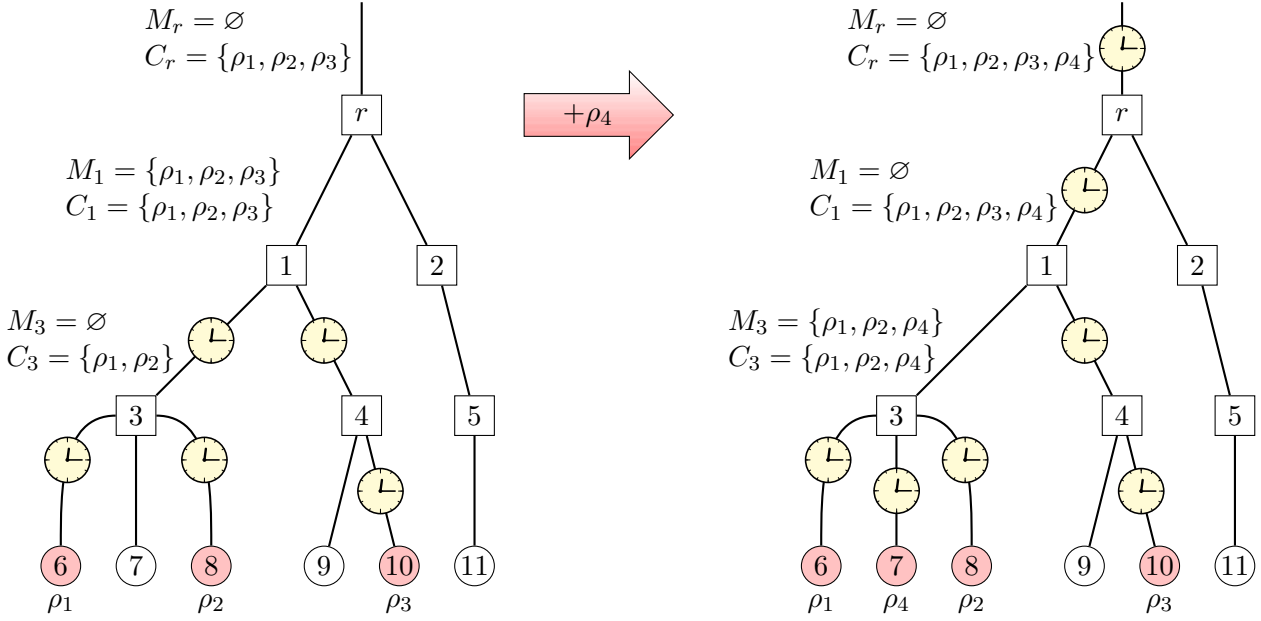
\begin{figure*}[bt]
 \hspace{-0.5cm}
 \begin{tikzpicture}[scale=1,baseline={(0,0)}]
	 \tikzstyle{vertex}=[draw, rectangle,minimum size=15pt,inner sep=0pt]
	 \tikzstyle{leaf} = [draw, circle, minimum size = 15pt,inner sep=0pt]
	 \tikzstyle{active leaf} = [leaf, fill=red!24]
	 \tikzstyle{saturated edge} = [draw,line width=5pt,-,red!50]
	 \tikzstyle{edge} = [draw,thick,-,black]
	 \node[vertex] (v0) at (4,6) {$r$};
	 \node[vertex] (v1) at (3,4) {$1$};
	 \node[vertex] (v2) at (5,4) {$2$};
	 \node[vertex] (v3) at (1,2) {$3$};
	 \node[vertex] (v4) at (4,2) {$4$};
	 \node[vertex] (v5) at (5.5,2) {$5$};
	 \node[active leaf] (v6) at (0,0) {$6$};
	 \node[leaf] (v7) at (1,0) {$7$};
	 \node[active leaf] (v8) at (2,0) {$8$};
	 \node[leaf] (v9) at (3.5,0) {$9$};
	 \node[active leaf] (v10) at (4.5,0) {$10$};
	 \node[leaf] (v11) at (5.5,0) {$11$};
	 \draw[edge] (v0) -- (v1) to  node[midway] {\ourclock} (v3) to[out=0,in=90]  node[midway] {\ourclock} (v8);
	 \draw[edge] (v0) -- (v2) -- (v5) -- (v11);
	 \draw[edge] (v1) to  node[midway] {\ourclock} (v4) -- (v9);
	 \draw[edge] (v4) to  node[midway] {\ourclock} (v10);
	 \draw[edge] (v6) to[out=90,in=180]  node[midway] {\ourclock} (v3);
	 \draw[edge] (v3) -- (v7);
     \draw[edge] (v0) -- (4,7.5);
	 
	 \node (labelv3) at ($(v3) + (-0.6,1)$) {\begin{tabular}{l}$M_3=\varnothing$ \\ $C_3=\{\rho_1,\rho_2\}$\end{tabular}};
	 \node (labelv1) at ($(v1) + (-1.5,1)$) {\begin{tabular}{l}$M_1=\{\rho_1, \rho_2,\rho_3\}$ \\ $C_1=\{\rho_1, \rho_2,\rho_3\}$\end{tabular}};
	 \node (labelv3) at ($(v0) + (-1.5,1)$) {\begin{tabular}{l}$M_r=\varnothing$ \\ $C_r=\{\rho_1, \rho_2,\rho_3\}$\end{tabular}};
	 
	 \node (labelv6) at ($(v6) + (0,-0.5)$) {$\rho_1$};
	 \node (labelv8) at ($(v8) + (0,-0.5)$) {$\rho_2$};
	 \node (labelv10) at ($(v10) + (0,-0.5)$) {$\rho_3$};
	 
 \end{tikzpicture}
~
\begin{tikzpicture}[scale=1,baseline={(0,0)}]
\node[draw, single arrow,minimum height=1.7cm,minimum width=0.5cm,single arrow head extend=.3cm,top color=transparent!0, bottom color=red!50] at (-1,6) {$\quad + \rho_4\quad $};
 \end{tikzpicture}
~
 \begin{tikzpicture}[scale=1,baseline={(0,0)}]
	 \tikzstyle{vertex}=[draw, rectangle,minimum size=15pt,inner sep=0pt]
	 \tikzstyle{leaf} = [draw, circle, minimum size = 15pt,inner sep=0pt]
	 \tikzstyle{active leaf} = [leaf, fill=red!24]
	 \tikzstyle{saturated edge} = [draw,line width=5pt,-,red!50]
	 \tikzstyle{edge} = [draw,thick,-,black]
	 \node[vertex] (v0) at (4,6) {$r$};
	 \node[vertex] (v1) at (3,4) {$1$};
	 \node[vertex] (v2) at (5,4) {$2$};
	 \node[vertex] (v3) at (1,2) {$3$};
	 \node[vertex] (v4) at (4,2) {$4$};
	 \node[vertex] (v5) at (5.5,2) {$5$};
	 \node[active leaf] (v6) at (0,0) {$6$};
	 \node[active leaf] (v7) at (1,0) {$7$};
	 \node[active leaf] (v8) at (2,0) {$8$};
	 \node[leaf] (v9) at (3.5,0) {$9$};
	 \node[active leaf] (v10) at (4.5,0) {$10$};
	 \node[leaf] (v11) at (5.5,0) {$11$};
	 \draw[edge] (v0) to  node[midway] {\ourclock} (v1) -- (v3) to[out=0,in=90]  node[midway] {\ourclock} (v8);
	 \draw[edge] (v0) -- (v2) -- (v5) -- (v11);
	 \draw[edge] (v1) to  node[midway] {\ourclock} (v4) -- (v9);
	 \draw[edge] (v4) to  node[midway] {\ourclock} (v10);
	 \draw[edge] (v6) to[out=90,in=180]  node[midway] {\ourclock} (v3);
	 \draw[edge] (v3) to  node[midway] {\ourclock} (v7);
     \draw[edge] (v0) to  node[midway] {\ourclock} (4,7.5);
	 
	 \node (labelv3) at ($(v3) + (-0.8,1)$) {\begin{tabular}{l}$M_3=\{\rho_1, \rho_2,\rho_4\}$ \\ $C_3=\{\rho_1, \rho_2,\rho_4\}$\end{tabular}};
	 \node (labelv1) at ($(v1) + (-1.5,1)$) {\begin{tabular}{l}$M_1=\varnothing$ \\ $C_1=\{\rho_1, \rho_2,\rho_3,\rho_4\}$\end{tabular}};
	 \node (labelv3) at ($(v0) + (-2,1)$) {\begin{tabular}{l}$M_r=\varnothing$ \\ $C_r=\{\rho_1, \rho_2,\rho_3,\rho_4\}$\end{tabular}};

	 \node (labelv6) at ($(v6) + (0,-0.5)$) {$\rho_1$};
	 \node (labelv8) at ($(v8) + (0,-0.5)$) {$\rho_2$};
	 \node (labelv10) at ($(v10) + (0,-0.5)$) {$\rho_3$};
	 \node (labelv7) at ($(v7) + (0,-0.5)$) {$\rho_4$};
	 
 \end{tikzpicture}
 \caption{Example of the algorithm on tree metrics for $k=3$. 
 Squares represent internal nodes of the tree and circles represent leaves. 
 Red leaves represent the leaves with an open request. Yellow clocks on edges mean that the timers are running. On the left side, there are three open requests. At this time, the supporting request set at node $1$ (the requests are potentially matched across node $1$) is $\{\rho_1,\rho_2,\rho_3\}$, hence no timer is running on top of node $1$. When $\rho_4$ arrives, we recalculate $M_v$ for every node in the tree, which is shown on the right side. Now, $\rho_1,\rho_2$ and $\rho_4$ are potentially matched across node $3$, and then $M_1$ becomes empty and the timer on the edge between nodes $1$ and $3$ is not running while the timer on top of node $1$ is running. }
 \label{fig:algorithm_example}
 \end{figure*}

For the following analysis, we define $\finalTimer_{v}$ as the final value of $\Timer_{v}$ after serving all requests. For every timer $\Timer_v$, we call each $i$-th time interval during which this timer starts being active and until it is consumed a \emph{phase} $\phi_{v,i}$. Note that in the final phase, a timer might not be consumed. In the following analysis, we will always consider that a phase is defined with respect to the online algorithm. 

The goal of this section is to prove the following theorem:

\setcounter{theorem}{6} 
\begin{theorem}\label{thm:correctnessTheorem_app}
For any given $k$ and any request sequence $R$, algorithm $\A$ achieves a competitive ratio $O(\log n)$, i.e.,
\begin{equation*}
\Expect[\Cost_{\A}(R)] \le O(\log n)\cdot \Cost_{\OPT}(R). 
\end{equation*}
\end{theorem}
For simplicity of the algorithm and the analysis, we do not intend to optimize the constant inside the competitive ratio in this paper. This constant depends on $k$ and varies with the defined metric space.

We will start the proof of Theorem \ref{thm:correctnessTheorem_app} by bounding the space cost of $\A$ through the final values $\finalTimer_v$ of the timers. 

\begin{lemma}\label{lm:space}
 The space cost $\sCost_{\A}(R)$ is at most $O(k) \cdot \sum_v \finalTimer_v$. 
\end{lemma}
\begin{proof} 
 Observe that algorithm $\A$ is only able to use an edge $e_v$ for matching if this edge has become inactive. The corresponding timer $\Timer_{v}$ must be an integral multiple of $w_v$ and the value that can be consumed from this timer is equal to $w_v$. Let $\#e_v$ denote the number of times an edge $e_v$ is used to match requests on top of $v$. This number can be bounded from above by using the final value of the corresponding timer and the weight of the corresponding edge: $\#e_v \leq \lfloor \finalTimer_v / w_v \rfloor$. 
 By considering all edges in the tree, the space cost of the algorithm $\A$ can be bounded by $$k\cdot \sum_{v\neq r} w_v \cdot \#e_v \le k\cdot \sum_{v\neq r} w_v \lfloor \finalTimer_v / w_v \rfloor \le k\cdot \sum_{v\neq r} \finalTimer_v$$ 
\end{proof}

In order to bound the time cost of the algorithm, we need to understand how the timers influence the total waiting time. 
We therefore need to introduce a so-called \emph{supporting request set}, which is defined on vertices of $T$ in a bottom-up manner. 
At time $t$, the set of supporting requests contains requests which can potentially be matched across $v$ and is denoted by $M_{v}(t)$. 
For all leaves $v\in \Leaves$ we set $M_{v}(t) = \emptyset$. 
Given an internal vertex $v\in T-\Leaves$, $M_{v}(t)$ can be computed by solving the following constraints: 
\begin{itemize}

\item Only open requests can be matched, and if some requests can potentially be matched under $v$ then we do not let them be supporting requests in $M_v(t)$:\\
 $M_{v}(t)\subseteq \left( \Active_{v}(t) - \bigcup_{u\in \Descendant(v)} M_{u}(t)\right)$
 
\item As many requests as possible should be inside the set of supporting requests:\\
$|M_{v}(t)|$ is an integral multiple of $k$, and $\left|\Active_{v}(t) -M_{v}(t) - \bigcup_{u\in \Descendant(v)} M_{u}(t)\right| < k$

\end{itemize}
Note that there may be multiple possible solutions satisfying the conditions above. For the rest of the analysis, we arbitrarily choose one of them as $M_v$. Figure~\ref{fig:algorithm_example} shows a sample constellation of the sets $M_{v}(t)$ and $C_{v}(t)$ for three selected vertices and how these sets change with a new arriving request. 

Using this concept, we can now also bound the time cost of $\ALG$ on any request sequence:

\begin{lemma}\label{lm:time}
 The time cost $\tCost_{\A}(R)$ is at most $O(k)\cdot\sum_v \finalTimer_v$. 
\end{lemma}
\begin{proof}

We first consider every node $v\in T-\{r\}$ such that $e_v$ is active and define $S(v):=\{\text{open requests } $ $\rho | \text{ all edges between } v \text{ and } \Location(\rho) \text{ are inactive.}\}$. 
We are sure that $|S(v)|$ is less than $k$, otherwise, these requests must be matched by $\A$. 
Since $\Timer_v$ is active, we can use it to count the delay cost of the requests in $S(v)$, up to a factor of at most $k$. 

In a similar way, we can define $S(r)$. In the following, we analyze two cases describing whether $S(r)$ contains all the current open requests or not. 

If $S(r)$ does not contain all the current open requests, then we know that there must exist at least one $v$ such that $\Timer_v$ is increasing. We can use this timer to count the delay cost of the requests in $S(r)$, up to a factor at most $k$. 

If $S(r)$ contains all the current open requests, one can set $w_r=\infty$, as $r$ has no parent node and therefore the respective timer $\Timer_r$ never becomes inactive.  Note that 
$\bigcup_{v\in T-\{r\}-\Leaves} M_v(t) = \bigcup_{u\in \Descendant(r)} M_{u}(t)$ since the supporting request sets of the leaves are empty by definition and the supporting requests of the root node are not counted. Consider the set of requests that can only be matched across the root $( \Active_{r}(t) - \bigcup_{u\in \Descendant(r)} M_{u}(t))$. The requests in this set will always increase the timer $\Timer_r$, since the timer of the root does not become inactive.
Observe that the timers of the root node are running at the same time for both, the online and the offline algorithms, since $\Active_r(t) \equiv \adv{\Active}_r(t) \pmod{k}$. That is, the timer for the offline algorithm will also always increase, such that this case cannot increase the competitive ratio and can therefore be ignored.

It is not difficult to see that $\bigcup_v S(v)$ is all the current open requests. Therefore, $$\tCost_{\A}(R)\le k \cdot \sum_v \finalTimer_v.$$ 
\end{proof}

\paragraph{Timers and Adversary-timers}\label{sec:timer}
So far, we have bounded the time and space costs of the online algorithm in terms of $\sum_v \finalTimer_v$. 
In order to calculate the competitive ratio of $\ALG$, we need to find a lower bound for $\OPT$ which also depends on $\sum_v \finalTimer_v$. We therefore introduce adversary-timers. 
For every node $v$ we define an adversary timer $\adv{\Timer}_{v}$, which is initialized to $0$ and increases at a unit rate as long as $\adv{C}_v(t)\not\equiv 0 \mod k$. Note that this definition is analogous to the definition of the timer $\Timer_{v}$.
For any adversary-timer $\adv{\Timer}_{v}$, its final value is denoted by $\adv{\finalTimer}_{v}$. 
We will consider the difference between the timers of the online and the optimal algorithms with respect to the phases defined by the online algorithm $\ALG$. We therefore denote the time cost of $\OPT$ in the phase $\phi_{v,i}$ by $\adv{\finalTimer}_{v,i}$. 
The space cost of $\OPT$ in the phase $\phi_{v,i}$ is respectively denoted by $\adv{\sigma}_{v,i}$. 
We can say that an amount of $w_v$ is added to $\adv{\sigma}_{v,i}$ whenever $\OPT$ matches requests on top of $v$. 
We further define the final value of the adversary-timer at node $v$ as $\adv{\finalTimer}_v = \sum \adv{\finalTimer}_{v,i}$, the total value of all adversary-timers as $\adv{\finalTimer} = \sum_v \adv{\finalTimer}_v$, the total space cost at a node $v$ as $\adv{\sigma}_v = \sum_i \adv{\sigma}_{v,i}$ and, finally, the total space cost of $\OPT$ as $\adv{\sigma} = \sum_v \adv{\sigma}_v$.

In order to discuss the difference between the timers, we need to introduce a notation for the change of the timer with respect to a phase. We therefore denote the change of $\Timer_v$ over phase $\phi$ by $\finalTimer(\phi)$, the change of $\adv{\Timer}_{v}$ by $\adv{\finalTimer}(\phi)$ and the change of the space cost (on the tree metric) of the adversary over the same phase by $\adv{\sigma}(\phi)$. Further, we call the phase $\phi_{v,i}$ a $j$-phase, if at the beginning of phase $\phi_{v,i}$ for the difference of open requests in $\adv{\Active}_v$ and in $\Active_v$ holds $(|\adv{\Active}_v| - |\Active_v|)\equiv j\mod{k}$.

\begin{lemma}\label{lem:0-phase}
 If $\phi=\phi_{v,i}$ is a $0$-phase, then $\finalTimer(\phi) \le \adv{\finalTimer}(\phi)$ or $\finalTimer(\phi) \le \adv{\sigma}(\phi)$. 
\end{lemma}
\begin{proof}
 For a $0$-phase $\phi$ of some node $v$, there are two possibilities
 \begin{description}
  \item[Case 1:] There is no matching on top of $v$ by $\OPT$ in phase $\phi$.
  \item[Case 2:] There exists at least one matching on top of $v$ by $\OPT$ in phase $\phi$.
 \end{description}
 In the first case, $\adv{\Timer}_v$ increases when $\Timer_v$ increases, such that $\finalTimer(\phi) \le \adv{\finalTimer}(\phi)$ holds.
 In the second case, there is at least one matching on top of $v$ by $\OPT$. Since $\OPT$ has to incur space cost in order to match the requests, the space cost is at least $\adv{\sigma}(\phi) \ge w_v$. Since a phase of $\ALG$ ends when the timer $\phi$ reaches the weight of the corresponding edge $e_v$ and is consumed thereafter, we have $w_v \ge \finalTimer(\phi)$. For the second case we therefore get that $\adv{\sigma}(\phi) \ge \finalTimer(\phi)$.
\end{proof}

Next, we will introduce two observations which are needed for the rest of the analysis. The first observation results from the following idea: We can say that in any phase $\phi$ in which there is at least one matching on top of $v$, the inequality $\finalTimer(\phi) \le \adv{\sigma}(\phi)$ holds independent of the value of $j$. This observation follows by applying the same arguments as in the proof of Lemma \ref{lem:0-phase}. 

\begin{observation}\label{ob:ontop}
 If $\phi$ is a phase and there exists at least one matching on top of $v$ by $\OPT$ in phase $\phi$, 
 then $\finalTimer(\phi) \le \adv{\sigma}(\phi)$. 
\end{observation}

The next question is how to resolve the case of $j$-phases ($j\neq 0$) which do not have a matching on top of $v$. In these phases, it is possible that $\adv{\finalTimer}(\phi) $ is also very small. 
We assume without loss of generality that exactly $(-|\adv{\Active}_v| \mod{k})$ requests arrive at $\Leaves_v$ at the beginning of these phases, 
and no other request arrives until the end of the phase. We can make this assumption because only if $|\adv{\Active}_v| \equiv 0 \mod{k}$ the timer $\Timer_v$ is increasing and the adversary timer $\adv{\Timer_v}$ is paused. We call the corresponding phase a \emph{harmful} phase, because we would not be able to bound the cost of $\OPT$ if there were several consecutive harmful phases. Fortunately, the following observation excludes this possibility.

\begin{observation}\label{ob:1-phase}
$(|\Active_v| \mod{k})$ is decreasing during consecutive harmful phases.
\end{observation}
\begin{proof}
There are only two ways to increase $(|\Active_v| \mod{k})$: either through new requests in the subtree rooted at $v$, or by using more than $(|\Active_v| \mod{k})$ for a matching on top of $v$. 
However, in these consecutive harmful phases, since $|\adv{\Active_v}| \equiv 0 \mod{k}$, it follows that $(|\Active_v| \mod{k})$ can only be changed by a matching on top of $v$ in $\ALG$, and not through new requests. Besides, when matching on top of $v$ in $\ALG$, according to our algorithm, only less than $(|\Active_v| \mod{k})$ requests  in the subtree rooted at $v$ are allowed to be used. 
\end{proof}

By combining Lemma \ref{lem:0-phase} with the observations, we can derive a lower bound for the costs of $\OPT$ incurred on the node $v$:
\begin{lemma}\label{lm:compare}
For every $v\in T$, $\finalTimer_v \le O(k \cdot (\adv{\finalTimer}_v + \adv{\sigma}_v))$. 
\end{lemma}
\begin{proof}
Observation \ref{ob:1-phase} implies that in any $k$ consecutive phases there must exist at least one non-harmful phase.  
For any non-harmful phase $\phi$ holds that at least one of $\adv{\finalTimer}_v(\phi)$ or $\adv{\sigma}_v(\phi)$ is greater than $w_v$. 
Moreover, the first phase of any node $v\in T-\Leaves$ is a $0$-phase. Together with the above result this implies the statement in the lemma. 
\end{proof}

The above lemmas provide lower bounds for $\adv{\finalTimer}_v$ and $\adv{\sigma}_v$. In order to prove the main result, we need the following two lemmas, which connect $\adv{\sigma}_v$ to the space cost and $\adv{\finalTimer}_v$ to the time cost of the offline algorithm on general metrics respectively. We start by bounding the space cost:

\begin{lemma}\label{lm:adv_space}
 $\Expect[\sum_v \adv{\sigma}_v] \le O(\log n) \cdot \sCost_{\OPT}(R)$
\end{lemma}
\begin{proof}
  $\sum_v \adv{\sigma}_v$ is the space cost of $\OPT$ on the tree metric. 
  When we embed $H$-metric $\M=(V,d)$ into tree metric $(T,w)$ using Lemma \ref{lemma:treeEmbed}, the expected distortion is $O(\log n)$. Therefore, we get
 $$\Expect\left[\sum_v \adv{\sigma}_v\right] \le O(\log n) \cdot \sCost_{\OPT}$$ 
\end{proof}

We can derive an analogous result for the time cost of the optimal algorithm: 

\begin{lemma}\label{lm:adv_time}
 $\sum_v \adv{\finalTimer}_v \le O(\log n) \cdot \tCost_{\OPT}(R)$
\end{lemma}
\begin{proof}
 When there is an open request $\rho$ under $\OPT$, at most $|\Ancestors(\Location(\rho))|+1$ timers (from $\Location(\rho)$ to the root) count the time cost incurred by this request.   
 By Lemma \ref{lemma:treeEmbed}, the height of tree $T$ and $|\Ancestors(\Location(\rho))|+1$ are bounded by $O(\log n)$. From this the claim follows. 
\end{proof}

 The main result from Theorem \ref{thm:correctnessTheorem_app} can be proved by combining the lemmas presented in this section: 
\[\begin{array}{rclr}
 \Expect[\Cost_{\A}(R)] & = & \Expect[\sCost_{\A}(R)] + \Expect[\tCost_{\A}(R)] & \\
                        &\le& \Expect[\sum_v \finalTimer_v] + \Expect[2\sum_v \finalTimer_v] &\qquad (\text{Lemma } \ref{lm:space},\ref{lm:time})\\
                        &\le& O(\Expect[\adv{\finalTimer}_v + \adv{\sigma}_v])& \qquad (\text{Lemma } \ref{lm:compare})\\
                        &\le& O(\log n) \cdot \sCost_{\OPT}(R) + O(\log n) \cdot \tCost_{\OPT}(R) &\qquad (\text{Lemma } \ref{lm:adv_space},\ref{lm:adv_time})\\
                        & = & O(\log n) \cdot \Cost_{\OPT}(R) &
\end{array}\]

\end{document}